\documentclass[11pt,a4paper]{article}

\usepackage[utf8]{inputenc}
\usepackage[margin=1in]{geometry}
\usepackage{amsmath,amssymb,amsthm,mathrsfs}
\usepackage{tabularx}
\usepackage{listings}
\usepackage{microtype}
\usepackage[colorlinks=true,linkcolor=blue,citecolor=blue,urlcolor=blue]{hyperref}

\newtheorem{defn}{Definition}
\newtheorem{remark}{Remark}
\newtheorem{theorem}{Theorem}
\newtheorem{proposition}{Proposition}

\newcommand{\ModkP}{\mathbf{FMod_{\text{m}}P}}
\newcommand{\tu}[1]{\mathbf{#1}}
\newcommand{\fun}[1]{\mathsf{#1}}
\newcommand{\cc}[1]{\mathbf{#1}}
\newcommand{\ODEcmod}{\text{ODE}_{\text{cmod\text{-}m}}}
\newcommand{\coNP}{\mathbf{coNP}}
\newcommand{\FP}{\mathbf{FP}}
\newcommand{\AC}{\mathbf{FAC}}
\newcommand{\NC}{\mathbf{FNC}}
\newcommand{\Nat}{\mathbb{N}}
\newcommand{\sP}{\mathbf{\#P}}
\newcommand{\oP}{\mathbf{\oplus P}}
\newcommand{\TC}{\mathbf{FTC}}
\newcommand{\FPH}{\mathbf{FPH}}
\newcommand{\FCH}{\mathbf{FCH}}
\newcommand{\NP}{\mathbf{NP}}
\newcommand{\acc}{\mathtt{acc}}

\newcommand{\ACDL}{\mathbb{ACDL}}

\begin{document}

% --- Title & Authors ---
\title{\textbf{Towards a Characterization of Counting and Alternating Classes via Discrete Ordinary Differential Equations: Ongoing Research Report}}

\author{
  \textbf{Melissa Antonelli}\thanks{ORCID: 0009-0006-9072-4847. Email: \texttt{melissa.antonelli@uni-tuebingen.de}} \\
  \small Carl Friedrich von Weizs\"acker-Zentrum, Universit\"at T\"ubingen, Germany
  \and
  \textbf{Eduardo Skapinakis}\thanks{ORCID: 0000-0002-5729-1078. Email: \texttt{eduardo.skapinakis@uni-tuebingen.de}} \\
  \small Carl Friedrich von Weizs\"acker-Zentrum, Universit\"at T\"ubingen, Germany \\
  \small Center for Mathematics and Applications (NOVA Math), NOVA FCT, Portugal
}

\date{}

\maketitle

% --- Abstract ---
\begin{abstract}
This paper presents a high-level report on an ongoing project aiming to leverage implicit approaches based on discrete ordinary differential equations (ODEs) to study multiple complexity classes, even beyond small circuit and polynomial-time classes.
Stimulated by recent ODE-based characterizations of polynomial-time functions ($\mathbf{FP}$) and classes over the reals, the research project outlined here pushes this investigation further into counting and alternation. Specifically, we present a uniform framework, built upon a single base algebra and a unified family of schemas, where complexity levels, such as those of the polynomial and counting hierarchies, are captured simply by the nesting depth of ODE operators. 
Crucially, our approach starts from a base class much weaker than $\mathbf{FP}$, thus strengthening existing recursion-theoretic treatments and establishing a natural connection to descriptive complexity. 
Moreover, by isolating three elementary schemas, our framework makes the computational content of linearity restrictions completely transparent while extending ODE-based implicit complexity to previously unaddressed counting classes, such as $\oplus\mathbf{P}$.
More generally, this work establishes a clear bridge between differentiation and counting, offering a fresh perspective on the relationships between different complexity classes, which remains the object of ongoing and future research.
\end{abstract}


\vspace{0.5cm}

%----------------------------------------------------%
\section{Introduction}
%----------------------------------------------------%
In recent years, the use of ordinary differential equations (ODEs) as a framework to design algorithms and characterize various complexity classes has seen a resurgence of interest.
Given their inherent connections to analog computation, ODEs have been used to characterize complexity over the reals~\cite{BG16,BGP,blanc:23}, and, more recently, their discrete counterparts have also been shown to handle standard computational models~\cite{BournezDurand,MFCS24}.
This establishes the ODE framework as a promising and versatile tool for machine-independent complexity, the branch of complexity theory dedicated to characterizing functions without relying on specific machines or cost models (see e.g.~\cite{cobham:65,leivant:94,fagin:74,immerman:12,Buss,Girard}).

The seminal work introducing \emph{discrete} ODEs in complexity established a characterization for the class of poly-time computable functions ($\FP$)~\cite{BournezDurand}.
In a series of subsequent works~\cite{MFCS24,MFCS25,ICALP}, this approach has also proven fruitful for studying small circuit classes, revealing one of the key advantages of ODEs:
the interplay of multiple parameters (such as discrete vs. continuous inputs, the choice of functions to derive along, the syntactical form of the equations, and constraints on their calls) enables a flexible characterization of (circuit) computational resources.
The research project delineated here aims to push this investigation further by exploring the expressive power of discrete ODEs in a different direction, namely by considering classes \emph{above} $\FP$.
This demonstrates the versatility of this mathematical tool by showing that they can naturally capture resources related to different machine models, such as nondeterministic and counting machines, thereby providing a comprehensive and new perspective on complexity.

In this paper, we outline the framework needed to pursue this \emph{systematic} and \emph{uniform} investigation of complexity via discrete ODEs, highlighting the advantages of this approach and the core strategies used to achieve it.
In Section~\ref{sec:ODEs}, we cursorily recall the foundations of the ODE-based approach in complexity (for full details, see e.g.~\cite{BournezDurand,TCS}).
In Section~\ref{sec:key}, we present the key ideas underlying our unified investigation of the polynomial and counting hierarchies.
In Section~\ref{sec:overview}, we provide a \emph{high-level} overview of the primary characterizations established in this project, presenting a global map of the discrete characterizations investigated so far, along with a couple of case studies.
We conclude by pointing out open directions for this ongoing work.

%----------------------------------------------------%
\section{A Brief Review of Discrete ODEs in Complexity}\label{sec:ODEs}
%----------------------------------------------------%

As is standard, we define ODEs as expressions of the form $\frac{\partial f(x,\tu y)}{\partial x}=h(x,\tu y, f(x,\tu y))$, where $\frac{\partial f(x,\tu y)}{\partial x}$ is the partial derivative of $f(x,\tu y)$ with respect to $x$ (being $\tu y$ fixed).
When an initial condition $f(0,\tu y)=g(\tu y)$ is added, we obtain a so-called \emph{initial value problem} (IVP).
Here, we focus on the discrete counterpart of ODEs, such that the discrete derivative of $f(x)$ is defined as $f(x+1)-f(x)$.

While these objects are already sufficient to describe relevant classes of functions, such as primitive recursive ones~\cite{BournezDurand}, in order to capture \emph{complexity} classes two novel concepts are introduced into the standard setting of difference calculus:
\begin{itemize}
\itemsep0em
    \item[a.] the idea of deriving along functions, enabling to control the number of recursion steps;
    \item[b.] a syntactic form of the equation, enabling to control the object size.
\end{itemize}
These two simple features allow for the definition of schemas which do not impose explicit bounds on the recursion (as in~\cite{cobham:65}) or assign special roles to variables (as in~\cite{BellantoniCook,leivant:94a,leivant:94}).
Formally, (a.) is realized through the notion of $\lambda$-ODE.

\begin{defn}[$\lambda$-ODE schema]
    Given $g,\lambda, u$, the function $f$ is said to be defined by $\lambda$-ODE from $g$ and $u$ if it is the solution of the IVP with initial value $f(0,\tu y) = g(\tu y)$ and such that:
    \[
    \frac{\partial f(x,\tu y)}{\partial \lambda} = \frac{\partial f(x,\tu y)}{\partial \lambda (x,\tu y)} = u(x,\tu y, f(x,\tu y)),
    \]
    being a formal synonym of $f(x+1,\tu y)=f(x,\tu y)+ \big(\lambda(x+1,\tu y)-\lambda(x,\tu y)\big) \times u(x,\tu y,f(x,\tu y))$.
    When $\lambda$ is the length function $\ell$ (s.t.~$\ell(0)=0$ and $\ell(x) = \lceil \log_2(x+1)\rceil$) this is called length-ODE.
\end{defn}

\noindent
This ensures that the value of $f(x,\tu y)$ changes only when the value of $\lambda(x,\tu y)$ does, thereby linking the number of recursion steps directly to the growth rate of the function along which we are deriving. Specifically, $\ell$ changes only when the length of its input increases, leading to an approximately logarithmic number of steps.

To control function growth, we apply linearity-based syntactical restrictions; these ensure that the output of recursive calls is strictly bounded at each step, preventing growth explosions.

\begin{defn}[Linear $\lambda$-ODE]\label{def:ODE}
Given $g,\lambda,h$, the function $f$ is said to be defined by $\lambda$-ODE from $g$ and $h$ if it is the solution of the IVP with initial value $f(0,\tu y)=g(\tu y)$ and such that:
\[
\frac{\partial f(x,\tu y)}{\partial \lambda} = A(x,\tu y, f(x,\tu y), h(x,\tu y)) \times f(x,\tu y) + B(x,\tu y, f(x,\tu y), h(x,\tu y))
\]
where $A$ and $B$ are expressions over the signature $\{+,-,\times, \fun{sg}\}$ (being $\fun{sg}:\mathbb{Z}\to \mathbb{Z}$ the sign function over $\mathbb{Z}$ taking value 1 for $x>0$ and 0 otherwise) in which $f$ occurs only under the scope of the sign function. 
If no call to $f$ occurs in $A$ and $B$, the schema is said to be \emph{strict}.
\end{defn}

\noindent
Recall that when dealing with linear systems, $f(x,\tu y)=\sum^{x-1}_{u=-1} \prod^{x-1}_{t=u+1} (1+A(t,\tu y, h,f)) \times B(u,\tu y, h, f)$, with $\prod^{x-1}_x \kappa(x)=1$ and $B(-1, \cdot)=f(0,\tu y)$. For further details, see~\cite{BournezDurand}. 

So-called linear length ODE, $\ell$-ODE, is defined by putting these two features together.
This schema was the crucial novelty introduced to capture poly-time computation, offering the very first characterization of~$\FP$ in terms of $\mathbb{LDL}$, an ODE-based function algebra made of standard basic functions ($\fun{0}, \fun{1}, \pi^p_i, \ell, \fun{sg}, +, -, \times$) and closed under composition ($\circ$) and $\ell$-ODE.
Subsequently, it was shown that discrete ODEs also provide a natural tool for expressing core operations of circuit-based computation, such as basic right-shifting or left-shifting with bit addition. This led to multiple original characterizations for classes ranging from those computable by unbounded Boolean circuits of polynomial size and constant depth ($\AC^0$) to those of logarithmic depth ($\NC^1$), also including various counting gates~\cite{MFCS24,MFCS25,ICALP}; notably, most of these latter classes had hitherto lacked implicit characterizations. 
Recall in particular that $\AC^0$ was characterized by a function algebra called $\mathbb{ACDL}$, defined by standard basic functions and closed under composition and under a special schema allowing for iterated left-shifting, with possible bit addition (see~\cite{MFCS24,MFCS25}).

Motivated by the simplicity with which ODE schemas capture crucial operations like (sharply) bounded search and counting, our project aims to extend this approach to complexity beyond $\FP$, thereby simplifying and unifying the characterizations of these classes as well.

%----------------------------------------------------%
\section{Uniformly Reasoning about Computation Paths}\label{sec:key}
%----------------------------------------------------%
In this Section, we introduce the intuitive ideas at the basis of our proposal for a uniform characterization of multiple classes, including all levels in the polynomial and counting hierarchies.

In recursion theory, complexity classes related to notions of \textit{alternation} and \textit{counting} have been mostly captured by forms of \emph{tree recursion} with parameter substitution (see, e.g.~\cite{leivant:94,leivant:00,oitavem:08,dal:22,oitavem:22}) or restricted forms of \emph{primitive recursion} (see, e.g.~\cite{thompson:72,wagner86,ben:12}).
In both cases, one uses recursion schemes that mimic the traversal of a \textit{computation tree}: from its root to its leaves for tree recursion (on notation), and from the last leaf back to the root for primitive recursion.

This approach to recursion can be naturally viewed as a form of discrete differentiation.
Thus, our key strategy is to combine the recognition of accepting computation paths (already doable in $\AC^0$) with mechanisms for searching and counting over them, tools precisely offered by ODE-schemas (Def.~\ref{def:ODE}). 
More generally, this framework makes it possible to capture a wide spectrum of classes, ranging from small circuits to polynomial space, by relying on \emph{the same function algebra} and by introducing two simple modifications to the core defining schemas: (A) a restricted form of composition (Def.~\ref{def:limitedComp}) and (B) specific constraints on the linearity of ODEs.

\paragraph{A. Restricted composition.}
As is standard in recursion theory~\cite{cobham:65,CloteTakeuti}, our characterizations are in terms of function algebras.
However, in order to uniformly capture hierarchy levels, we rely on a special form of restricted composition, inspired by~\cite{oitavem:11,oitavem:22,dal:22}.
This addresses the fact that classes like $\cc{NP}$ are not known to be closed under complement, and thus also not under composition.
For example, let $\fun{NEG}(x):=1-x$ and suppose that $A\in\cc{NP}$. 
Then, while $1_B(x):=1_A(\fun{NEG}(x))$ is still in $\cc{NP}$, $1_{A^c}(x) = \fun{NEG}(1_A(x))$ is in $\cc{coNP}$.
The solution to remain in $\cc{NP}$ is thus to restrict the \textit{inner} functions that can be used in composition.

\begin{defn}\label{def:limitedComp}
Given a (base) algebra $\mathscr{B}$ and a set of operators $\mathcal{O}_1,\dots,\mathcal{O}_k$, the algebra $\mathscr{F} = [\mathscr{B}; \circ_0, \mathcal{O}_1, \dots, \mathcal{O}_k]$ is defined as the standard algebra $[\mathscr{B}; \circ, \mathcal{O}_1, \dots, \mathcal{O}_k]$ but with $f$ being obtained by composition, $f(x,\tu y)=h(\tu{g}(x,\tu y))$, only when all (inner) functions in $\tu g$ are in $\mathscr{B}$.
\end{defn}

If the base class $\mathscr{B}$ is itself defined via operators $\mathcal{O}_1',\dots,\mathcal{O}_m'$, these cannot be used in $[\mathscr{B}; \circ, \mathcal{O}_1, \dots, \mathcal{O}_k]$, i.e.,~one has access to functions in $\mathscr{B}$, but not to operators defining them. This restriction allows for fine-grained control over the \textit{nested} application of certain schemas. 

\paragraph{B. Some paradigmatic ODE schemas.}
We consider strict ODE-based schemas obtained by imposing various restrictions on the linearity of the defining equation. Specifically, when $A$ takes values only in $\{-1,0\}$ and $B=0$, the computation corresponds to universal bounded search; when $A=-B$ and $B$ takes values in $\{0,1\}$, it corresponds to existential bounded search; if $A=0$ and $B$ takes values in $\{0,1\}$, counting is obtained.

\begin{defn}\label{def:parODE}
   Let $g:\Nat^p \to \{0,1\}, k_1:\Nat^{p+1}\to \{-1,0\}$ and $k_2:\Nat^{p+1}\to \{0,1\}$.
   The function $f:\Nat^{p+1}\to \Nat$ defined as the solution of the IVP with initial value $f(0,\tu y)=g(\tu y)$ and such that:
   \begin{itemize}
   \itemsep0em
       \item $\frac{\partial f(x,\tu y)}{\partial x}=f(x,\tu y)\times k_1(x,\tu y)$ is said to be defined by ODE$_\wedge$
       \item $\frac{\partial f(x,\tu y)}{\partial x}=-f(x,\tu y)\times k_2(x,\tu y)+k_2(x,\tu y)$ is said to be defined by ODE$_\vee$
       \item $\frac{\partial f(x,\tu y)}{\partial x}=k_2(x,\tu y)$ is said to be defined by ODE$_{\#}$.
   \end{itemize}
\end{defn}

It is easy to check that solutions of the corresponding systems precisely correspond to the desired operations: universal bounded search, existential bounded search, and counting.

\paragraph{Recognizing accepting computation paths.}
Typical machine-independent characterizations of non-deterministic classes with polynomial resource restrictions ($\cc{NP}$, $\sP$, \dots) are defined having as the base class $\cc{P}$ or $\cc{FP}$.
Thus, having at our disposal ODE schemas which mimic the operations required to define these classes, a natural approach to characterize them would be to close the algebra of functions $\mathbb{LDL}$~\cite{BournezDurand} under these schemas.
Here we show that, in fact, we can use a much weaker class as our basis, namely~$\ACDL$~\cite{MFCS24}.
Indeed, while characterizations in~\cite{oitavem:22,dal:22} rely on schemas to represent the non-deterministic bits generated by a non-deterministic Turing machine (NTM) $M$, we instead use them to generate the entire computation of $M$.

Then, what we need is simply a mechanism allowing us to check whether the generated words correspond to accepting computations of $M$, which can be done already in $\AC^0$ (similar formulations of this result have been proved in~\cite[Th.~3.1]{vollmer:96} and~\cite[Th.~7.8]{immerman:12}).
The core intuition behind this result is that, under a suitable encoding, one can interpret a computation path of a machine $M$ as a binary word. For each machine $M$, we can define the set
\[
\mathtt{acc}_{M} = \{(z,x) \in \mathbb{N}^2 : z \text{ is an accepting computation of } M \text{ on input } x\}.
\]

\begin{theorem}\label{th:enc}
For any (N)TM $M$ and input $x$, there is an $\cc{FAC}^0$ function deciding $\mathtt{acc}_{M}$. 
\end{theorem}

This approach partially strengthens the state of the art~\cite{vollmer:96,oitavem:11,oitavem:22,dal:22} and establishes a clear(er) bridge between the ODE-based approach and descriptive complexity, where it is also {\sf FO} ($=\cc{AC}^0$) which serves as the basis for the characterization of complexity classes~\cite{fagin:74,immerman:89,immerman:12}.

%----------------------------------------------------%
\section{Capturing Classes Beyond $\FP$ via Discrete ODEs: an Overview}\label{sec:overview}
%----------------------------------------------------%
Simply combining $\ACDL$ and restricted composition with ODE$_\wedge$ or ODE$_\vee$ is enough to capture the base levels of $\FPH$. For example, 

\begin{theorem}\label{th:NP}
    $\AC^0 \cup \cc{NP} = [\ACDL; \circ_0, \emph{ODE}_\vee]$.
\end{theorem}
\begin{proof}[Proof idea]
Since $\ACDL=\AC^0$~\cite{MFCS24}, we focus on restricted composition and ODE$_\vee$.
$(\supseteq)$ The closure property of $\circ_0$ follows directly from the definition of restricted composition (see also~\cite{oitavem:11}).
When $f$ is defined by ODE$_\vee$ from $g$ and $k$, by Def.~\ref{def:parODE}, $f(x,\tu y)=1$ iff $g(\tu y)=1$ or there is a $u\in \{0,\dots, x-1\}$ s.t.~$k(u,\tu y)=1$.
By the induction hypothesis, $g$ and $k$ are in $\NP$, so that computing $f(x,\tu y)$ amounts to computing a bounded existential quantifier over $\NP$ functions.
Since $\NP$ is closed under bounded existential quantifiers, $f$ can be computed in $\NP$.
$(\subseteq)$ Let $M$ be a poly-time $\NP$-machine $M$ and let $V(z,x)\in \ACDL$ be the function deciding the set $\acc_M$.
Since $M$ runs in poly-time, there is a polynomial $p:\mathbb{N} \to \mathbb{N}$ s.t.~for any $x$, the encoding of an accepting computation of $M(x)$ has \emph{precisely} length $p(\ell(x))$. Let $r:\mathbb{N} \to \mathbb{N}$ be s.t.~$\ell(r(x))=p(\ell(x))$.
The details of the encoding can be defined in various ways (see~\cite{wagner86}).
We define $f$ by ODE$_\vee$ from 0 and $V$, being $V(z,x)=0$ for any word $z$ of size smaller than $r(x)$, $f(1 \# r(x), x)=1$ iff there is a word $z$ of size $\ell(r(x))$ (a possible computation path of $M$) s.t.~$V(z,x)=1$, i.e.,~for which $z$ represents an accepting computation of $M$ on $x$.
\end{proof}

\noindent
The class $\AC^0\, \cup\, \cc{coNP}$ is captured in an analogous way, by considering ODE$_\wedge$, and subsequent levels of $\cc{PH}$ are characterized by iteratively considering the previous level as the basic function class.

A similar approach extends to the study of the counting classes~\cite{Valiant} and hierarchy~\cite{wagner86} ($\cc{FCH}$) and its levels, where the crucial schema is instead ODE$_{\#}$, i.e.,~a schema that allows one to ``count'' paths. 
Intuitively, completeness direction shows that if $f \in \sP$, i.e.~if there is a poly-time NTM $M$ s.t., for every $x$, $f(x) = \#\acc_{M,x}$, we can use ODE$_{\#}$ to define $F(r(x),x)=f(x)$ using $g(z)=0$ and $V(z,x)\in\ACDL$ which returns $1$ iff $z$ is an accepting computation of $M$ on input $x$.
Notably, the flexibility of this approach allows us to straightforwardly capture other counting classes never characterized before, such as $\oP$ or $\ModkP$, by considering, again, $\ACDL$, $\circ_0$ but different-counting ODE schemas (e.g., a schema \emph{counting modulo 2}).

Other simple restrictions on linearity lead to ODEs with significant computational power. Although this investigation is still in progress, at least two of them seem especially worth mentioning. 
Endowing $\ACDL$ with an ODE schema that differentiates with respect to $x$ and such that $A=-1$ (or is restricted to $\{0,1\}$) is enough to capture $\cc{FPSPACE}$.
Observe that this schema closely mirrors the one defining $\NC^1$, which instead differentiates with respect to $\ell$; a connection which deserves further investigation.
Moving to derivations along $\ell$, an alternative characterization for $\FP$ can be obtained by considering a schema featuring a generalized form of linearity, where $A=1$ and $B$ takes values in $\{0,1\}$ rather than $\ell$-ODE~\cite{BournezDurand}.
The relationship between these two characterizations is currently being explored.

%----------------------------------------------------%
\section{Perspectives}
%----------------------------------------------------%

The paper presents a high-level overview of an ongoing project that aims to advance the study of complexity through the perspective of ODEs and clarify the power of these schemas, for instance, to capture alternating hierarchies. 
Developing clear ODE instruments and systematically investigating these characterizations is the precise focus of the research sketched here.

This investigation is relevant for multiple reasons. 
First, it further exemplifies how flexible ODEs can be: besides circuits and dynamical systems, ODEs are able to capture resources associated with non-determinism, relativized computation and counting machines. 
Additionally, while most of the classes mentioned here have already been characterized using standard tools available in recursion theory, the ODE perspective offers for the first time a \emph{unique} mathematical object capable of characterizing diverse aspects of computation, spanning both the discrete and continuous domains. 
To our knowledge, such universality is not offered by any previously developed tools in recursion theory.
The versatility of this framework is made possible by the variety of features defining ODEs (e.g., strict vs.~non-strict, functions to derive along, linearity, etc.). 
Clarifying the relationship between these constraints remains an open question for future research, potentially bridging other approaches to implicit complexity,~e.g.,~by comparing how computational resources are captured by constraints on ODEs vs. logical features~\cite{fagin:74,immerman:89,immerman:12}.

Finally, far from being just another way to rewrite recursion schemas, this framework is capable of straightforwardly capturing classes never characterized before in recursion theory, as shown in~\cite{ICALP} and for~$\oplus\mathbf{P}$. We are thus confident that these tools can also provide original characterizations for other relevant classes (e.g.,~involving randomness~\cite{Gill77} or monotonicity~\cite{grigni:92}) and offer new means to study the relationships between them.

\vspace{0.5cm}

\begin{table}[ht]
\centering
\small
\begin{tabular}{|c | c |c |c || c| c | c | c |}
 \hline
 &\multicolumn{3}{c||}{Deriving along $x$} & \multicolumn{3}{c|}{Deriving along $\ell$} \\
 \cline{2-7} & $\AC^0 \cup \mathbf{PH}$ & $\FCH$ & $\cc{FPSPACE}$ 
  & $\TC^0$ & $\cc{FNC}^1$ & $\FP$ \\ 
 \hline\hline
 lev.~1 & $\circ_0$ ; $-f \times k_2 + k_2$  & $\circ_0$ ; $k_2$ & {} & {} & {} & $f\times A + B$ \\ 
 \hline
 lev.~$n$ & $\circ_n$ ; $-f  \times k_2 + k_2$ & $\circ_n$ ; $k_2$ & $-f + B$ & $k_2$ & $-f + B$ & {} \\
 \hline
 union & $-f  \times k_2 + k_2$ & $k_2$ & {} & {} & {} & $f+K$ \\
 \hline
\end{tabular}
\caption{\footnotesize Summary of the defining features of ODE schemas characterizing different classes and hierarchy levels. Characterizations are obtained by closing $\ACDL$ or the preceding hierarchy level (denoted by $\circ_n$) under (possibly restricted, $\circ_0$) composition and the given ODE schema. Whenever composition is not restricted, $\circ$ is omitted. We follow the convention that $k_1:\Nat \to \{-1,0\}$, $k_2:\Nat \to \{0,1\}$, and $B$ (resp., $K$) is a (resp., generalized) $\fun{sg}$-polynomial expression, s.t.~$f$ only occurs under the scope of $\fun{sg}$ (resp.,~takes values in $\{0,1\}$).}
\end{table}

% --- Bibliography ---
% Un-comment and add your bib file if submitting to arXiv with BibTeX
 \bibliographystyle{plain}
 \bibliography{bib.bib}

\newpage
%----------------------------------------------------%
%----------------------------------------------------%
\appendix
%----------------------------------------------------%
%----------------------------------------------------%
\section{Checking accepting computations}\label{app:accept}
%----------------------------------------------------%
%----------------------------------------------------%
We prove that checking if a word represent an accepting computation can be done in $\AC^0$.
This result has been proven e.g. in~\cite[Th.~3.1]{vollmer:96} and~\cite[Th.~7.8]{immerman:12}), but we give here a simple a self-contained proof. For this, we must first fix an encoding of Turing machines (TMs).
Our machine model is a standard (one-tape) TM over the tape alphabet $\Sigma= \{\triangleright, \square, 1, 0\}$, where $\triangleright$ denotes ``beginning of the tape marker'' and $\square$ is for ``blank'' (for details, see e.g.~\cite{arora:09}).

\begin{defn}
    A Turing machine is a tuple $M=(Q, \delta, q_0,q_a, q_r)$ where:
    \begin{itemize}
        \itemsep0em
        \item $Q$ is a finite set of non-final states 
        \item $q_0\in Q$ is the starting state;
        \item $q_a \not\in Q$ is the \emph{only} accepting state;
        \item $q_r \not\in Q$ is the \emph{only} rejecting state (the last two being the halting states);
        \item $\delta \subseteq (Q\times \Sigma) \times \big((Q\cup \{q_a,q_r\}) \times (\Sigma \times \{L,S,R\})\big)$ is the transition relation;
        \item $L,S$ and $R$ represent the head-moves ``left'', ``stay'', and ``right''.
        %; whenever $\delta(q,\sigma_1, \dots, \sigma_k)=(q', \tau_1, d_1, \dots, \tau_k, d_k)$ if $\delta_j=\triangle$, then $\tau_j=\triangle$ and $d_j=R$.
    \end{itemize}
\end{defn}
\noindent
When $\delta$ is a function and not just a relation, $M$ is a deterministic TM (DTM); otherwise it is a non-deterministic TM (NTM).

%When $M$ is an NTM, multiple accepting computations can be associated to each word $x$

\begin{theorem}
For any TM $M$, the set $\acc_{M}$, containing the pairs $(z,x)$ such that $z$ encodes an accepting computations of $M$ on input $x$, is in $\AC^0$.
\end{theorem}
\begin{proof}

%\begin{definition}[Accepting computation]
    \begin{sloppypar}
    Let $M=(Q,\delta, q_0, q_a, q_r)$, where $|Q|=m$ is the set of states, run in time bounded by a polynomial $r$.
    Let $x\in \{0,1\}^n$ and $r=r(n)$.
    A possible computation of $M(x)$ can be described by a sequence of $(r+2)(r+1)$ words
    $$
    \underbrace{c_0^0 c_0^1 \dots c_0^r s_0}_{\emph{initial config.}} 
    %content of the tape and initial state}}
    c_1^0 c_1^1 \dots c_1^r s_1 \dots 
    \underbrace{c_r^0 c_r^1 \dots c_r^r s_r}_{\emph{final config.}}
    $$
    where each $c_i^j \in \{0,1\}$ contains the symbol in square $j\in \{1,\dots, r-1\}$ at time step $i\in\{0,\dots, r\}$ and a bit indicating whether the tape head is in this cell.
    Each word $s_i$ represents the respective state at time  $i$.
    Such sequence is an \emph{accepting computation} if: 
    \begin{itemize}
        \itemsep0em
        \item for all $i\in\{0,\dots, r\}$, $(c_i^0)_0 = \triangleright$ and $(c^r_i)_0 = \square$, as the machine alters the content of the cells in positions between 1 and $r-1$ only;
        \item $s_0=q_0$, $c^0_0=(\triangleright, 1)$, for $j\in \{1,\dots, n\}$, $c_1^j = (x_{j-1}, 0)$, for all $j\in \{n, \dots, r\}$, $c_0^j=(\square, 0)$ (initial configuration);
        \item $s_r=q_a$ (accepting configuration);
        \item for any time step $i$ and tape position $j$, with $i,j\in \{0,\dots, r-1\}$, $(c_i^{j-1}, c_i^j, c_i^{j+1}, s_i) \vdash_M (c^{j-1}_{i+1}, c_{i+1}^j, c_{i+1}^{j+1}, s_{i+1})$, i.e.~the sequence respects the transition relation of $M$.\footnote{In fact, having $c_i^0=\triangleright$ and $c_i^r=\square$, for all $i\in \{0,\dots, r\}$, prevents one from having to distinguish the edges of the tape.}
    \end{itemize}
    The relation $\vdash_M$ is the predicate which, given $s_i$ (i.e.~the state at time $i$) and $c_i^{j-1}, c_i^j, c_i^{j+1}$ (i.e.~configurations of cells from $j-1$ to $j+1$, which encode the symbols of the tape cells and if the tape head is there) asserts whether the word $c_{i+1}^j$ correctly represents the new symbol and head position of $M$ and $s_{i+1}$ corresponds to the new state at step $i+1$.
    If $M$ is deterministic, only one pair $c_{i+1}^j$ and $s_{i+1}$ satisfies the predicate.
%\end{definition}
\end{sloppypar}

Each state $q$ of $M$ can be encoded as a distinct binary word $\lceil q\rceil$ of size $\lceil$log$_2(m+1)\rceil$ and the symbols of $\Sigma=\{0, 1, \triangleright,\square\}$ as 00, 01, 10, and 11, resp.
This means that each $c_i^j$ can be represented by a binary word of size 3 (the extra bit signaling the tape head) and the sequence above has size $k(n)=\lceil$log$_2(m+1)\rceil\times (r(n)+1)+3 \times (r(n)+1)^2$.
Thus, the relation $\vdash_M$ can be represented as a predicate $\delta_M$ on binary words.

\begin{sloppypar}
    %Let $M$ be an $\NP$-machine with running time bounded by $r$. Let $x\in\{0,1\}^n$ be the input to $M$.
    %
    %Recall that 0, 1, $\triangleright$ and $\square$ are encoded as 00, 01, 10, and 11, resp. 
    %
    %Let $k(n)=\lceil$log$_2(|Q|+1)\rceil \times (r(n)+1)+3 \times (r(n)+1)^2$.%, which is a polynomial in $n$.
    %
    For $0\leq i,j\leq r$, let $z^j_i$ denote the binary word encoding $c_i^j$ as described above and let $z^{r+1}_i$ encode $s_i$.
    Then, $z$ is an accepting computation if the following expression holds:
    % TODO MEL : Maybe, this sentence is not super nice
    \begin{itemize}
        \itemsep0em
        \item conditions on the first and last symbols are expressed as: $\bigwedge^r_{i=0} z_i^0(2) z_i^0(1) = 10 \wedge \bigwedge^r_{i=0} z_i^r(2) z_i^r(1)=11$;
        \item conditions on initial and final configurations can be defined as: $z_0^0 = 101 \wedge \bigwedge^n_{j=1} z_0^j = 0x_j0 \wedge \bigwedge^r_{j=n+1} z_0^j=110 \wedge z_0^{r+1} = \lceil q_0\rceil \wedge z_r^{r+1} = \lceil q_a\rceil$;
        \item conditions on the transition functions are given by: $\bigwedge^{r-1}_{i=0} \bigwedge^{r-1}_{j=1} \delta_M(z_i^{j-1}, z_i^j, z_i^{j+1}, z_{i+1}^j, z_{i+1}^{r+1})$.
    \end{itemize}
    The conjunction of those can be checked by a circuit family of constant depth, so $\acc_M\in\AC^0$.
\end{sloppypar}
\end{proof}

%----------------------------------------------------%
%----------------------------------------------------%
%----------------------------------------------------%
%----------------------------------------------------%
\section{The computational content of ODE schemas}
%----------------------------------------------------%
%----------------------------------------------------%
\begin{proposition}\label{prop:ODEs}
    Let $k:\Nat^{p+1}\to \{0,1\}$. If $f(x,\tu y)$ is defined by:
    \begin{enumerate}
    \itemsep0em
        \item ODE$_\wedge$ from $g(\tu y):= 1$ and $k_1(x,\tu y):=k(x,\tu y)-1$ (with $k$ taking values in $\{0,1\}$), then $f(x,\tu y)=1$ iff for every $i\leq x-1$, $k(i,\tu y)=1$;
        \item ODE$_\vee$ from $g(x,\tu y):=0$ and $k_2(x,\tu y)=k(x,\tu y)$, then $f(x,\tu y)=1$ iff there is $i\leq x-1$ s.t.~$k_2(i,\tu y)=1$;
        \item ODE$_{\#}$ from $g(x,\tu y):=0$ and $k_2(x,\tu y)=k(x,\tu y)$, then $f(x,\tu y)=\sum^{x-1}_{u=0}k_2(u,\tu y)$.
    \end{enumerate}
\end{proposition}
\begin{proof}
The proof follows by Def.~\ref{def:ODE} and solutions of linear systems.
Regarding item $1$, it is easy to see that a function defined by ODE$_\wedge$ from $g(\tu y)=1$ and $k_1(x,\tu y)=k(x,\tu y)-1$ performs the following computation 
$$
f(x,\tu y)=\prod^{x-1}_{i=-1}(1+k_1(i, \tu y))
$$
with $k_1(-1,\tu y)=g(\tu y)$.
Clearly, this product is $1$ if ($g(\tu y)=1$ and) for all $t\in \{0,\dots, x-1\}$, $k_1(t, \tu y)=0$, that is for any $t$, $k(t,\tu y)=1$.
%It is then always possible to define a function $k_F(x,\tu y)$ taking values in $\{0,1\}$ to be such that~$k_F(x,\tu y)=1$ iff $k_1(x,\tu y)=0$ so that $f(x,\tu y)=1$ iff $g(\tu y)=1$ and for any $t\in \{0,\dots, x-1\}$, $k_F(t,\tu y)=1$.
\\
Similarly, regarding item $2$, a function defined by ODE$_\vee$ from $g$ and $k_2$ computes 
\begin{align*}
f(x,\tu y) &=\sum^{x-1}_{u=-1} \prod^{x-1}_{t=u+1}(1-k_2(t,\tu y)) \times k_2(u,\tu y)
\end{align*}
where  $k_2(-1, \tu y)=g(\tu y)$.
Since $k_2$ takes only values in $\{0,1\}$, $f(x,\tu y)= 1$ iff $g(\tu y)=1$ or there is at least a $t\in \{0,\dots, x-1\}$ such that~$k_2(t,\tu y)=1$, i.e.~such that $k(t,\tu y)=1$.
\\
Finally, regarding item $3$, a function defined by ODE$_{\#}$ from $g$ and $k$ computes precisely
\begin{align*}
f(x,\tu y) &=\sum^{x-1}_{u=-1} k_2(u,\tu y),
\end{align*}
that is counting $g(\tu y)=1$ and $t\in\{1,\dots, x-1\}$ such that $k(t,\tu y)=1$.
\end{proof}
%----------------------------------------------------%
%----------------------------------------------------%
%----------------------------------------------------%
%----------------------------------------------------%
\section{The polynomial time hierarchy}\label{app:PH}
%----------------------------------------------------%
%----------------------------------------------------%
In this Section, we generalize Th.~\ref{th:NP} and provide the first ODE-based characterizations for the levels of $\cc{PH}$.

\begin{defn}\label{def:complexity_classes}
    %Let $n\ge 1$. The polynomial and counting hierarchy are inductively defined respectively as:
    The polynomial hierarchy is defined as $\cc{PH} := \bigcup_{n\geq0}\Sigma_n^p = \bigcup_{n\geq0}\Pi_n^p$, where, for $n\ge 1$:
    \begin{align*}
    \Sigma^p_0 &:= \Pi^p_0 := \cc{P} \quad \quad \quad \Sigma^p_{n+1} := \NP^{\Sigma^p_n} \quad \quad \quad \Pi^p_{n+1} := \coNP^{\Sigma^p_n}
    \end{align*}
\end{defn}

\noindent
The characterization of these classes relies on ODE$_\vee$ and ODE$_\wedge$ (resp.), which, as seen, intuitively allows one to perform existential and universal bounded search over \emph{all possible} computation paths of a NTM, checking the acceptance of which is in $\AC^0$.

First, let us consider the dual of Th.~\ref{th:NP}:

\begin{theorem}\label{th:coNP}
$\AC^0 \cup \coNP = [\ACDL; \circ_0, \emph{ODE}_\wedge]$.
\end{theorem}
\begin{proof}
$(\supseteq)$ Follows from $\ACDL=\AC^0$~\cite{MFCS24} and the closure of $\coNP$ under ODE$_\wedge$ (bounded universal quantification) and restricted composition.
$(\subseteq)$ 
To show that computation by a $\coNP$-machine can be simulated in our class we proceed similarly to Th.~\ref{th:NP}.
For a $\coNP$-machine $M$, let $V'$ be the $\ACDL$ function such that $V'(z,x) = 1$ iff $z$ represents an accepting computation of $M$ on input $x$, or does not correctly represent a computation of $M$ on input $x$ (this is required as we want the scheme to be ``neutral'' with respect to words not representing computations).
The theorem follows by simply 
defining $f$ by ODE$_\wedge$ from $1$ and $V'(z,x)-1$.
\end{proof}

The characterizations of $\AC^0 \cup \NP$ and $\AC^0\cup \coNP$ can be naturally generalized to uniformly capture all levels in $\cc{PH}$.
Following~\cite{oitavem:22}, we use the notation $\cc{C} \equiv [\mathscr{B}; \mathcal{O}_1, \dots, \mathcal{O}_k]$ to mean that the complexity class $\cc{C}$ is characterized precisely by those functions in $[\mathscr{B}; \mathcal{O}_1, \dots, \mathcal{O}_k]$ which output only values in $\{0,1\}$.

\begin{theorem}\label{th:PHk}
Let $\mathbb{SDL}_0:=\mathbb{PDL}_{0}:=\ACDL$. For any $n\ge 0$,
\begin{align*}
    &\Sigma^p_{n+1} \equiv  \mathbb{SDL}_{n+1} := [\mathbb{SDL}_n; \circ_0, \emph{ODE}_\vee] = [\mathbb{PDL}_n; \circ_0, \emph{ODE}_\vee]&\\
    &\Pi^p_{n+1} \equiv \mathbb{PDL}_{n+1} := [\mathbb{SDL}_n; \circ_0, \emph{ODE}_\wedge] = [\mathbb{PDL}_n; \circ_0, \emph{ODE}_\wedge] & \\
    &\cc{PH} \equiv \mathbb{SDL} := [\ACDL; \circ, \emph{ODE}_\vee] = \mathbb{PDL} := [\ACDL; \circ, \emph{ODE}_\wedge]&
\end{align*}
\end{theorem}
\begin{proof}
Similarly to~\cite{oitavem:22} the proof is by induction on the structure of the new algebras, following the ideas of Th.~\ref{th:NP} and~\ref{th:coNP}.
%\textcolor{red}{(see App.~\ref{app:PH})}. 
We describe the proof for the algebras $\mathbb{SDL}_{n}$, as this already contains all relevant ingredients.
$(\subseteq)$
Assume by induction that $\mathbb{SDL}_{n} \equiv \Sigma_n^p$.
In defining $\mathbb{SDL}_{n+1}:=[\mathbb{SDL}_n; \circ_0, \text{ODE}_\vee]$, the restriction on composition (Def. \ref{def:limitedComp}) allows functions in $\mathbb{SDL}_n$ to be freely composed. Thus, we can use $\fun{cosg}$ ($= 1-\fun{sg})$ to obtain any characteristic functions in $\Pi_n^p$.
Thus, in $\mathbb{SDL}_{n+1}$ we can apply ODE$_\vee$ to characteristic functions of sets in $\Pi_n^p$, so we can perform existential search on $\Pi_n$ predicates, obtaining the class $\Sigma_{n+1}^p$ by~\cite[Th 8.8(e)]{balcazar:88}.
$(\supseteq)$
Here one has to work with two simultaneous induction hypothesis: that $\mathbb{SDL}_n\equiv \Sigma_n^p$, and that the closure of $\mathbb{SDL}_n$ under composition is contained in $\cc{FP}^{\Sigma_n^p}$, both of which hold for $\mathbb{SDL}_1$ due to Th.~\ref{th:NP}.
For the algebra $\mathbb{SDL}_{n+1}$, just as described in the previous inclusion, one has access to not only $\mathbb{SDL}_n$, but also its closure under composition, which by induction is contained in $\cc{FP}^{\Sigma_n^p}$.
Thus, since $\Sigma_{n+1}^p$ is closed under bounded existential search over $\cc{FP}^{\Sigma_n^p}$ functions \cite[Prop. $8.1$(f)]{balcazar:88}, and under restricted composition \cite[Prop. 4]{oitavem:22}, it follows that the $\{0,1\}$ functions in $\mathbb{SDL}_n$ characterize at most $\Sigma_{n+1}^p$.
Thus, its closure under composition is contained in $\cc{FP}^{\Sigma_n^p}$ as well.
\end{proof}

\noindent
Regarding the functions (not only predicates) characterized by $\mathbb{SDL}_{n}$ and $\mathbb{PDL}_{n}$, we conjecture that these correspond to bounded query classes over $\Sigma_n^p$ (see~\cite{buhrman:99}).
%----------------------------------------------------%
%----------------------------------------------------%
%----------------------------------------------------%
%----------------------------------------------------%
\section{The counting hierarchy}\label{app}
%----------------------------------------------------%
%----------------------------------------------------%
In this Section, we present the characterizations of the function version of the counting hierarchy. 
This hierarchy can be defined using relativized poly-time counting machines, which are NTMs returning the number of their accepting computation paths, written in binary (see~\cite{Valiant,wagner86}).

\begin{defn}\label{def:complexity_classes}
    The class $\sP$ consists of the functions computable by poly-time counting machines. 
    %Let $n\ge 1$. The polynomial and counting hierarchy are inductively defined respectively as:
    The counting hierarchy is defined as $\cc{FCH} := \bigcup_{n\geq0}\cc{FCH}_n$, where, for $n\ge 1$:
    \begin{align*}
    \cc{FCH}_0 &:= \cc{FP} \quad \quad \quad \quad \; \;
    \cc{FCH}_{n} := \sP^{\cc{FCH}_{n-1}}.
    \end{align*}
\end{defn}

\begin{theorem}\label{th:CH}\label{th:sP}
Let $\mathbb{CDL}_{0}:=\ACDL = \AC^0$. For any $n\ge 0$, 
\begin{align*}
    &\FCH_{n+1} = \mathbb{CDL}_{n+1} := [\mathbb{CDL}_{n}; \circ_0, \emph{ODE}_\#]&\\
    &\cc{FCH} = \mathbb{CDL} := [\ACDL,\circ, \emph{ODE}_{\#}].&
\end{align*}
\end{theorem}
\begin{proof}
\begin{sloppypar}
First, consider the basic level $\FCH_1=\sP$.
    $(\supseteq)$ Since $\ACDL=\AC^0$~\cite{MFCS24}, the inclusion follows from the fact that $\sP$ is closed under ODE$_{\#}$~\cite[Th. 4.9]{wagner86} and restricted composition \cite[Lem. 9]{dal:22}.
    $(\subseteq)$ 
    If $f \in \sP$, then there is a poly-time NTM $M$ that, for every input $x$, $f(x) = \#\acc_{M,x}$.
    Let $V\in\ACDL$ be the function which, given $z$ and $x$, returns $1$ iff $z$ is an accepting computation of $M$ on input $x$.
    Define $F$ by ODE$_{\#}$ from $V(z,x)$, with $g(z)=0$. Then, there is a polynomial (on size) $r$ such that $F(r(x),x) = f(x)$, so $f\in\mathbb{CDL}_1$.

    Regarding the levels $\FCH_{n+1}$, the proof follows by induction, similarly to \cite{wagner86,vollmer:96,dal:22}.
    $(\supseteq)$ We rely on the closure of each level $\FCH_{n+1}$ under ODE$_{\#}$~\cite[Theorem 4.9]{wagner86} and restricted composition~\cite[Th. 17]{dal:22}.
    $(\subseteq)$ We use~\cite[Lem. 3.2]{wagner86}: each function in $\cc{FCH}_{n+1}$ can be written as $f(x) = \sum_{u=0}^{p(x)}h'(h(u,x))$, for $p,h'\in\cc{FP}$, with $h':\mathbb{N}\to\{0,1\}$ and $h\in\cc{FCH}_n$.
    Since $\cc{FP}\subseteq\sP$, for $n\geq 1$ we have that $p$ and $h'$ are in $\cc{FCH}_{n}$ so in $\cc{FCH}_{n+1}$. We conclude by construct $f(x) = \sum_{u=0}^{p(x)}h'(h(u,x))$ via ODE$_{\#}$ and restricted composition.
\end{sloppypar}
\end{proof}
%----------------------------------------------------%
%----------------------------------------------------%
%----------------------------------------------------%
%----------------------------------------------------%
\section{Modulo counting}\label{sec:modulo_counting}
%----------------------------------------------------%
%----------------------------------------------------%
The most common modulo class is $\oplus\cc{P}$, which can be characterized by NTMs accepting an input iff an odd number of their computation paths accept. 
Equivalently it can be seen as the set of functions returning the number of accepting computation paths of a NTM, modulo $2$.
This definition extends naturally to more general $\cc{FMod_{\text{m}}P}$, the class of functions computing the number of accepting computation paths of a NTM modulo $m$.

\begin{defn}
Let $m>1$. A function $f$ is  said to be defined by \emph{$\ODEcmod$} from $g: \Nat^p \to \{0,\dots, m-1\}$ and $k:\Nat^{p+1} \to \{0,1\}$ if it is the solution of the IVP with initial value $f(x,\tu y)=g(\tu y)$ and such that:
%        $$
%        \frac{\partial f(x,\tu y)}{\partial x} = -k_2(x,\tu y) \times n \times \Big\lfloor \frac{f(x,\tu y)}{n}\Big\rfloor + k_2(x,\tu y).
%        $$
\begin{align*}
    %f(0,\tu y) &= g(\tu y) \\
    \frac{\partial f(x,\tu y)}{\partial x} &= -m \times \bigg\lfloor \frac{f(x,\tu y)+k(x,\tu y)}{m}\bigg\rfloor + k(x,\tu y),
\end{align*}
\end{defn}

\begin{theorem} %\label{th:count}
For any $m>1$, $\cc{FMod_{\text{m}}P}\cup \AC^0=[\ACDL; \circ_0,$ \emph{$\ODEcmod$}$].$ 
\end{theorem}
\begin{proof}
$(\supseteq)$ All function and schemas of $\ACDL$ are in $\AC^0$~\cite{MFCS24}.
Suppose that $f$ is defined by ODE$_{\text{cmod-m}}$.
Observe that, for any $t\in \{0,\dots, x-1\}$, if $k(t+1, \tu y)=0$, then $f(t+1,\tu y)=f(t, \tu y)$; if $k(t+1,\tu y)=1$, then there are two possible cases:
$$
f(t+1,\tu y) = \begin{cases}
f(t,\tu y) + 1  \quad &\text{if } f(t,\tu y) < m - 1 \\
f(t, \tu y) - m \times \bigg\lfloor \frac{f(t,\tu y)+1}{m}\bigg\rfloor + 1 = 0 \quad &\text{if } f(t,\tu y) = m-1
\end{cases}
$$
Thus, $f(x,\tu y)=\sum^{x-1}_{i=-1}k(i,\tu y) \; \textsc{mod}\; m$, with $k(-1,\tu y)=g(\tu y)$.
Since, by construction, $k$ is in $\ACDL$ (Def.~\ref{def:limitedComp}) and our machine can count modulo $m$, this computation is doable in $\ModkP$.
$(\subseteq)$ As for Prop.~\ref{th:sP}, if $f\in \ModkP$, then there is a poly-time NTM $M$ returning the number of accepting paths modulo $m$.
This can be computed via ODE$_{\text{cmod-m}}$ from $k=V$, checking whether the encoded path is an accepting one.
\end{proof}
%----------------------------------------------------%
%----------------------------------------------------%
%----------------------------------------------------%
%----------------------------------------------------%
\section{Polynomial space}\label{sec:poly_space}
%----------------------------------------------------%
%----------------------------------------------------%
In order to capture $\cc{FPSPACE}$, we pass through a characterization for this class given by Clote. First, recall that a function $f$ is said to be defined by Concatenation Recursion on Notation (CRN) from $g$ and $k$ if $f(0,\tu y)=g(\tu y)$ and $f(s_i(x), \tu y) =s_{k_{i}(x,\tu y)}(f(x,\tu y))$~\cite{Clote1990}. 
Moreover, a function $f$ is said to be defined by $k$-Bounded Recursion ($k$-BR) from $g$ and $h$ if $f(0,\tu y)=g(\tu y)$ and $f(x+1,\tu y)=h(x,\tu y, f(x,\tu y))$, provided that $f(x,\tu y)\leq k$~\cite{Clote1990}. %See \cite{clote:99} for more details.

\begin{theorem}[\cite{clote:99}]\label{th:clote}
For $k\ge 4$, $\cc{PSPACE} \equiv [0, \ell, \fun{s}_0, \fun{s}_1, \fun{BIT}, \#,  \pi^p_i; \circ, \emph{CRN}, k\emph{-BR}]$.
\end{theorem}

Additionally, in order to characterize this class we will use so called $\ell$-ODE$_1$, a schema first introduced in~\cite{MFCS24} to characterize $\ACDL$ and whose computation intuitively correspond to iteratively left shifting the binary representation of a given number, each time possibly adding one in the last position.
For full details, see~\cite{MFCS24,MFCS25}.

\begin{defn}[$\ell$-ODE$_1$ schema]
    Given $g:\Nat^p \to \Nat$ and $k:\Nat^{p+2}\to \{0,1\}$, the function $f$ is said to be defined by $\ell$-ODE$_1$ from $g$ and $k$ if it is the solution of the IVP such that:
    \begin{align*}
        f(0,\tu y) &= g(\tu y) \\
        \frac{\partial f(x,\tu y)}{\partial \ell} &= f(x,\tu y)+k(x,\tu y).
    \end{align*}
\end{defn}

\noindent
Together with basic function in $\mathscr{B}$ (see Th.~\ref{th:PSPACE} below) this schema is enough to capture CRN.

To capture this class, we introduce a \emph{non-strict} ODE schema obtained deriving along $x$ and with a restricted linearity such that $A$ to take values in $\{-1,0\}$.

\begin{defn}[bODE schema]
    For $g:\Nat^p\to \Nat$, $k:\Nat^{p+1} \to \{0,1\}$ and $h:\Nat^{p+1} \to \Nat$, the function $f:\Nat^{p+1}\to \Nat$ is said to be \emph{bODE$^{\uparrow}$} definable from $g,k$ and $h$ if it is the solution of the IVP with initial value $f(0,\tu y)=g(\tu y)$ and such that:
    $$
    \frac{\partial f(x,\tu y)}{\partial x}=- k(x,\tu y) \times f(x,\tu y) + B(x,\tu y, h(x,\tu y), f(x,\tu y)),
    $$
    where $B$ is a $\fun{sg}$-polynomial expression (and $f$ occurs in it only under the scope of the sign function in it).
    If for any $x,\tu y$, $k(x,\tu y)=1$, then we call this schema \emph{bODE}.
\end{defn}

\begin{theorem}\label{th:PSPACE}
Let $\mathscr{B}= [\fun{0}, \fun{1}, \fun{sg}, \fun{BIT}, \ell, \times, -, \div 2, \#, \pi^p_i]$. Then,
$$
\cc{FPSPACE} = [\mathscr{B}; \circ, \emph{$\ell$-ODE$_1$}, \emph{bODE}] = [\mathscr{B}; \circ, \emph{$\ell$-ODE$_1$}, \emph{bODE}^\uparrow].
$$
\end{theorem}
\begin{proof}
    $(\subseteq)$ We start by showing that $[\mathscr{B}; \circ, \ell$-ODE$_1$, bODE] captures $\cc{PSPACE}$. By Th.~\ref{th:clote}, it is enough to show that, if $f$ is defined by $k$-BR from $g$ and $h$ (for $k\ge 4$), then it can be rewritten in $[\mathscr{B}; \circ, \ell$-ODE$_1$, bODE$]$.
For simplicity, let $k=4$. Then, let us consider $f_{br}(x,\tu y)$ defined as the solution of the IVP below:
\begin{align*}
    f_{br}(0,\tu y) &= g(\tu y) \\
    \frac{\partial f_{br}(x,\tu y)}{\partial x} &= - f_{br}(x,\tu y) + \sum^{4}_{i=0} h (i, \tu y) \times \fun{sg}(f_{br}(x,\tu y) - (i+1)) \times \fun{cosg}(f_{br}(x,\tu y)- i)
\end{align*}
where $h$ is precisely the function defining $k$-BR. 
Clearly, this is an instance of bODE and, since $\sum^{4}_{i=0}h(i,\tu y)\times \fun{sg}(f_{br}(x,\tu y)-(i+1)) \times \fun{cosg}(f_{br}(x,\tu y)-i)$ is simply a convoluted way to compute $\sum^{4}_{i=0}h(i,\tu y)\times (f_{br}(x,\tu y)=i)$, $f(x,\tu y)=f_{br}(x, \tu y)$.
Since our class contains $\cc{PSPACE}$ and can simulate a form of CRN (this is doable already by $\ACDL$, see~\cite{MFCS24}), it  contains $\cc{FPSPACE}$.
Let $F\in\cc{FPSPACE}$, such that the output size of $F$ is bounded by some polynomial $p:\mathbb{N}\to\mathbb{N}$.
Then, for $j\in\{0,1\}$, the sets $B_{F}^j:=\{(x,i)\,:\,\fun{BIT}(F(x),i) = j\}$ are decidable by $\cc{PSPACE}$ functions $H_0(x,i)$ and $H_1(x,i)$, which by the inclusion above are in our class.
Thus, for each \textit{size} $m\leq p(|x|)$, given an input $x$, we may compute $H_0(x,m)$ and $H_1(x,m)$, which will yield one of three possible answers for the $m$-th bit of $F(x)$: it is $0$, it is $1$, or $|F(x)|<m$.
Since $\ell$-ODE$_1$ allows one to simulate CRN, we can construct $F(x)$ by simply concatenating all its bits.

$(\supseteq)$ Since all functions in $\ACDL$ are already in $\cc{FPSPACE}$, we just have to show that $\cc{FPSPACE}$ is closed under bODE$^\uparrow$. Suppose that $f$ is defined by \emph{bODE}$^\uparrow$.
By Def.~\ref{def:ODE}, $f(x,\tu y)=B\big(x,\tu y, h(x,\tu y), f(x-1,\tu y)\big)$, where $f(x-1,\tu y)$ only occurs under the scope of the sign function.
Additionally, since $B$ is a $\fun{sg}$-polynomial expression, $f$ occurs in $B$ a fixed number of times for each recursive call, say $m$; that is, for any $t\in \{0,\dots, x-1\}$, $B(t,\tu y, h(t,\tu y), f(t,\tu y)) = B(t,\tu y, h(t,\tu y), \fun{sg}(B_1(t,\tu y, f(t-1,\tu y))), \dots, \fun{sg}(B_m(t,\tu y, f(t-1,\tu y))))$, where for any $i\in\{1,\dots, m\}$, $B_i$ is a $\fun{sg}$-polynomial expression.
Then, in order to compute $f(x,\tu y)$, we have to:
\begin{itemize}
    \itemsep0em
    \item first, compute $g(\tu y)$, which is in $\cc{FPSPACE}$ for hypothesis;
    \item then, at each step $t\in \{0,\dots,x-1\}$, we consider the value obtained in the previous call, say $z$, and, for any $i\in\{0,\dots, m\}$, we compute $\fun{sg}(B_i(z,t,\tu y))$; 
    this is doable in $\cc{FPSPACE}$. We store the result, a sequence $S(t)= b_0, ... b_{m-1} $ of $m$ bits (corresponding to $\fun{sg}(B_1(t,\tu y, z)), \dots, \fun{sg}(B_m(t,\tu y, z))$ respectively) in the memory;
    \item finally, we remove $z$ from the memory and overwrite with $B(t,\tu y, h(t,\tu y), S(t))$, computed using the results previously obtained and stored in the memory in the previous step; here, $h$ is in $\cc{FPSPACE}$ for hypothesis and $B$ is a $\fun{sg}$-polynomial expression.
\end{itemize}
Overall this yields a $\cc{FPSPACE}$ algorithm.
\end{proof}
%----------------------------------------------------%
%----------------------------------------------------%
%----------------------------------------------------%
%----------------------------------------------------%
\section{Polynomial time}\label{sec:poly_time}
%----------------------------------------------------%
%----------------------------------------------------%
In this Section, we introduce a new ODE-based characterizations for $\FP$.
To capture this class we consider the \emph{non-strict} version of $\ell$-ODE$_1$, the defining schema of $\AC^0$.
Notably, also $\ell$-ODE$^1$ is obtained deriving along $\ell$.

%and obtained deriving along $\ell$
%

\begin{defn}[$\ell$-ODE$^1$ schema]
    Let $g:\Nat^p\to \Nat$ and $k:\Nat^{p+1}\to \{0,1\}$. The function $f:\Nat^{p+1}\to \Nat$ is said to be \emph{$\ell$-ODE$^1$} definable from $g$ and $h$ if it is the solution of the IVP with initial value $f(0,\tu y)=g(\tu y)$ and such that:
    $$
    \frac{\partial f(x,\tu y)}{\partial \ell} = f(x,\tu y) + k(x,\tu y, f(x,\tu y)).
    $$
    %, where~$k$ occurs in it only under the scope of $\fun{sg}$.
\end{defn}

\noindent
In other words, $\ell$-ODE$_1$ is a very special, but generalized case of $\ell$-ODE \cite{BournezDurand}, where $A=1$ and $B$ is a \emph{generalized} form of $\fun{sg}$-polynomial expression and takes only values in $\{0,1\}$ (or, purely syntactically, is an expression occurring under the scope of the sign function).
In terms of recursion schemas, we show that just as $\ell$-ODE$_1$ captures CRN~\cite{Clote1990}, the computation performed by $\ell$-ODE$^1$ corresponds to Full Concatenation Recursion on Notation (FCRN), which, even being a very restricted form of BRN, still gives a characterization of $\cc{FP}$.
Recall that a function $f$ is defined by (F)CRN from $g$ and $h_i$ if $f(0,\tu y)=g(\tu y)$ and $f(s_i(x), \tu y) =s_{k_{i}(x,\tu y)}(f(x,\tu y))$ (resp., $f(s_i(x), \tu y) =s_{k_{i}(x,\tu y, f(x,\tu y))}(f(x,\tu y))$), $i\in \{0,1\}$ and $h_i:\Nat^p \to \{0,1\}$~\cite{Ishihara}.
% , in the same way . , where a
%Additionally, computation performed by $\ell$-ODE$^1$ captures to full concatenation recursion on notation (FCRN)~\cite{Ishihara}, in the same way as $\ell$-ODE$_1$ captures concatenation recursion on notation (CRN)~\cite{Clote1990}, where a function $f$ is defined by (F)CRN from $g$ and $k$ if $f(0,\tu y)=g(\tu y)$ and $f(s_i(x), \tu y) =s_{k_{i}(x,\tu y)}(f(x,\tu y))$ (resp., $f(s_i(x), \tu y) =s_{k_{i}(x,\tu y, f(x,\tu y))}(f(x,\tu y))$), $i\in \{0,1\}$ and $h:\Nat^p \to \{0,1\}$.

This schema is the key to characterize $\FP$.
The simplest proof follows an indirect path: on the one hand, $\ell$-ODE$^1$ is simply a special case of BRN, ensuring the closure property holds; on the other hand, (F)CRN can be rewritten using our schema, which yields completeness~\cite{Ishihara}.

%for this schema, which can be seen as nothing but a special case of BRN (see~\ref{app:FPSpace}).
%Based on this observation and on~\cite{Ishihara} a simple (indirect) characterization proof for poly-time computable functions can be provided (see App.~\ref{app:FPSpace}).
%
%Whereas the first is used to characterize $\AC^0$~\cite{Clote1990}, the second characterizes $\cc{FP}$~\cite{Ishihara}.
%Then, $\ACDL$ combined with $\ell$-ODE$^1$ and full composition, is sufficient to capture $\FP$. 

\begin{theorem}[\cite{Ishihara}]\label{th:Ishihara}
    $\FP = [\fun{0}, \fun{s}_0, \fun{s}_1, \fun{mod2}, \fun{msp}, \#, \pi^i_p; \circ, \emph{FCRN}]$
\end{theorem}

\noindent
The following remark also helps simplify the proof.

\begin{remark}\label{remark:FP}
    $[\ACDL; \circ, \ell\emph{-ODE}^1]  = [\fun{0}, \fun{1}, \ell, \fun{sg}, +, -, \div 2, \#, \pi^p_i; \circ, \ell\emph{-ODE}^1]$
\end{remark}
\begin{proof}[Proof Sketch]
    Follows from the fact that basic functions are the same and $\ell$-ODE$_1$ is nothing but a restricted version (namely, the strict counterpart) of $\ell$-ODE$^1$ such that $f(x,\tu y)$ does not appear in $k$ at all.
\end{proof}

\noindent
Then, as mentioned, we prove Th.~\ref{th:FP} indirectly.

\begin{proposition}\label{prop:lODE}
    If $f$ is defined by $\ell$\emph{-ODE}$^1$ from $g$ and $h$ in $\FP$, then $f$ is in $\FP$ as well.
\end{proposition}
\begin{proof}[Proof idea]
    By Def.~\ref{def:ODE}, $f(x,\tu y)=\sum^{x-1}_{u=-1} 2^{x-(u+1)} \times k(u,\tu y, f(u,\tu y))$, where $k(-1, \cdot, \cdot)=g(\tu y)$.
    First,  $g(\tu y)$ is computed in polynomial time (for hypothesis). Then, for any $t\in\{0,\dots, x-1\}$, where $z$ the value obtained  ``so far'', we compute $k(t,\tu y,z)=b_t (\in \{0,1\})$ (again, doable in $\FP$ for hypothesis), and  concatenate the resulting bit to $z$.
\end{proof}

\begin{theorem}\label{th:FP}
$\FP = [\ACDL; \circ, \ell\emph{-ODE}^1]$.    
\end{theorem}
\begin{proof}
Actually, due to Remark~\ref{remark:FP}, we will prove
$$
\FP = [\fun{0}, \fun{1}, \ell, \fun{sg}, +, -, \div 2, \#, \pi^p_i; \circ, \ell\text{-ODE}^1]
$$
    $(\subseteq)$ 
    All basic functions except $\fun{mod2}$ has been proved to be expressible already via $\ACDL$.
    In order to deal with $\fun{mod2}$, we consider $f_{mod2}$ defined as the solution of the IVP below:
    \begin{align*}
        f_{mod2}(0, y) &= \fun{BIT_0}(y) \\
        \frac{\partial f_{mod2}(x,y)}{\partial \ell} &= f_{mod2}(x,y) +  \fun{sg}(\fun{BIT}(f_{mod2}(x,y), 0)) \times \fun{cosg}(\fun{BIT}(x,\ell(x)+1)) \\
        &\quad \quad \quad \quad \quad \quad 
        +
        \fun{cosg}(\fun{BIT}(f_{mod2}(x,y), 0)) \times \fun{BIT}(x,\ell(x)+1)
    \end{align*}
    where $\fun{BIT_0}(z) = \fun{BIT}(z,0)$ and $\fun{BIT}(w, v)$ is the function intuitively returning the $v$-th bit of the binary representation of $w$; notably $\fun{BIT}(w,v)=\fun{bit}(w,\fun{bexp}(w,v)-1)$, where $\fun{bit}(w,v)=\fun{msp}(v,w)-2\times \fun{msp}(2v+1, w)$ and $\fun{bexp}(w,v)$ is the bounded exponential function that, for any $v\leq \ell(w)$ returns $2^v$ and definable using $\fun{msp}, \fun{if}$ and 2$^{\ell(\cdot)}$, together with $\ell$-ODE$^1$ (for further details see~\cite{MFCS24}).
    Clearly, this is an instance of $\ell$-ODE$^1$.
    We conclude by considering $\fun{BIT}(f_{mod2}(x,x),0)=\fun{mod2}(x)$.
    Finally, we deal with FCRN.
    Let $f$ be defined by FCRN from $g$ and $h_i$, with $i\in\{0,1\}$.
    This can be rewritten as the solution of the IVP below:
    \begin{align*}
        f_{fcrn}(0, y, \tu y) &= g(\tu y) \\
        \frac{\partial f_{fcrn}(x,y, \tu y)}{\partial \ell} %&= f_{fcrn}(x,\tu y) + h_{\fun{BIT}(\ell(y)-\ell(x)-1,y)}(x, \tu y, f_{fcrn}(x,\tu y)) \\
        &= f_{fcrn}(x,\tu y) + \fun{sg}\big(h_{\fun{BIT}(\ell(y)-\ell(x)-1,y)}(x, \tu y, f_{fcrn}(x,\tu y))\big)  \\
        &= f_{fcrn}(x,\tu y) + \fun{sg}(h_0(x,\tu y, f_{fcrn}(x,\tu y))) \times \fun{cosg}(\fun{BIT}(\ell(y)-\ell(x)-1, y)) \\
        &\quad \quad \quad \quad \quad \quad + \fun{sg}(h_1(x, \tu y, f_{fcrn}(x,\tu y))) \times \fun{sg}(\fun{BIT}(\ell(y)-\ell(x)-1,y)).
    \end{align*}
    This is clearly an instance of $\ell$-ODE$^1$.
    Then, $f(x,\tu y)=f_{fcrn}(x,x,\tu y)$, as desired.
    This is enough to obtain $[\fun{0}, \fun{s}_0, \fun{s}_1, \fun{mod2}, \fun{msp}, \#; \circ, \text{FCRN}]\subseteq [\fun{0}, \fun{1}, \ell, \fun{sg}, +, -, \div 2, \#, \pi^p_i; \circ, \ell\text{-ODE}^1]$, which due to Th.~\ref{th:Ishihara} concludes our proof.
    \\
    $(\supseteq)$ Since all basic functions are already in $\ACDL (=\AC^0 \subset \FP)$~\cite{MFCS24} and due to the closure property of $\ell$-ODE$^1$ (Prop.~\ref{prop:lODE}). In particular, soundness is indirectly ensured by the fact that $\ell$-ODE$^1$ is a special case of BRN: if $f$ is defined by $\ell$-ODE$^1$ from $g$ and $k$, then it can be easily defined by BRN from $g$ and a function $h$, which performs iterated concatenations of either 0 or 1 (based on the value of $k$ on the recursive call). Poly-time bound for computation by this function is thus ensured. %For the idea of a direct proof see Prop.~\ref{prop:lODE}.
\end{proof}

%\noindent
%Soundness is indirectly ensured by the fact that $\ell$-ODE$^1$ is a special case of BRN: if $f$ is defined by $\ell$-ODE$^1$ from $g$ and $k$, then it can be easily defined by BRN from $g$ and a function $h$, which performs iterated concatenations of either 0 or 1 (based on the value of $k$ on the recursive call). Poly-time bound for computation performed by this function is thus ensured.

%----------------------------------------------------%
%----------------------------------------------------%
%----------------------------------------------------%
%----------------------------------------------------%

\end{document}